\documentclass{article}
\usepackage{ijcai21}

\usepackage{times}
\usepackage{soul}
\usepackage{url}
\usepackage[utf8]{inputenc}
\usepackage[small]{caption}
\usepackage{graphicx}
\usepackage{amsmath}
\usepackage{amsthm}
\usepackage{booktabs}
\usepackage{mathtools}
\usepackage{amsfonts}
\usepackage{amssymb}
\usepackage{mathtools}
\usepackage{bbm}
\usepackage{bm}
\usepackage[ruled,vlined,linesnumbered]{algorithm2e}
\usepackage{subcaption}
\usepackage{stfloats}

\newtheorem{theorem}{Theorem}
\newtheorem{lemma}[theorem]{Lemma}

\newtheorem{claim}{Claim}
\newtheorem{fact}{Fact}
\newtheorem{assumption}{Assumption}
\newtheorem{example}{Example}
\theoremstyle{definition}
\newtheorem{definition}{Definition}

\theoremstyle{remark}

\newcommand{\etal}{\emph{et al.}}

\title{Dominant Resource Fairness with Meta-Types}

\author{
Steven Yin\footnote{Contact Author}\and
Shatian Wang\and
Lingyi Zhang\And
Christian Kroer\\
\affiliations
IEOR Department, Columbia University
\emails
\{sy2737, sw3219, lz2573, christian.kroer\}@columbia.edu
}

\begin{document}

\maketitle

\begin{abstract}
Inspired by the recent COVID-19 pandemic, we study a generalization of the multi-resource allocation problem with heterogeneous demands and Leontief utilities. Unlike existing settings, we allow each agent to specify requirements to only accept allocations from a subset of the total supply for each resource. These requirements can take form in location constraints (e.g. A hospital can only accept volunteers who live nearby due to commute limitations). This can also model a type of substitution effect where some agents need 1 unit of resource A \emph{or} B, both belonging to the same meta-type. But some agents specifically want A, and others specifically want B. We propose a new mechanism called Dominant Resource Fairness with Meta Types which determines the allocations by solving a small number of linear programs. The proposed method  satisfies Pareto optimality, envy-freeness, strategy-proofness, and a notion of sharing incentive for our setting. To the best of our knowledge, we are the first to study this problem formulation, which improved upon existing work by capturing more constraints that often arise in real life situations. Finally, we show numerically that our method scales better to large problems than alternative approaches. 
\end{abstract}

\section{Introduction}
The recent COVID-19 pandemic has brought forward many problems that
are particularly relevant to the operations research and computer science communities. Among them, an often overlooked problem is
the effective and fair allocation of resources, such as volunteer medical
workers, ventilators, and emergency field hospital beds.

There are several key challenges to the medical resource allocation
problem in the face of an infectious disease outbreak. First, utilities from
different types of resources are not additive nor linear. For example,
when there are enough nurses but not enough doctors, the marginal utility
of having one additional nurse on staff is very low.
Second, not all resources are accessible to all hospitals (referred to henceforth as \textit{accessibility constraints}). For instance, the home location of each volunteer medical
worker determines where she can commute to work; thus, she can only be
assigned to hospitals within her commutable radius. Third, hospitals have
different capacities (big medical centers versus small hospitals) and are in different
stress levels (hospitals in an epicenter versus the ones in rural areas with few cases), so they should
naturally be prioritized differently.

Another setting that has the above characteristics is the compute resource sharing problem
with sub-types. For example, suppose a compute server has
several compute nodes, and there are different types of GPU/CPU on the various
nodes (e.g. NVIDIA vs. AMD GPU, Intel vs. AMD CPU). Some users look for specific hardware configurations (e.g. accept only Intel CPU) while
others might be less selective (e.g. accept all CPU brands).

In this paper, we propose a new market mechanism that tackles the three
challenges outlined above and achieves desirable  properties
including Pareto optimality, envy-freeness, 
strategy-proofness, and sharing incentive. 
In our numerical experiments, we demonstrate that
compared to the Maximum Nash Welfare (MNW), and the Discrete MNW approach, our mechanism is cheaper to compute 
and enjoys theoretical properties that MNW approaches do not have.


\section{Related Work}
\label{sec:relatedwork}
Recently, a flurry of papers have come out of the operations
research, statistics, and computer science communities addressing
resource allocations in the aftermath of a pandemic. 
For instance,
Jalota \etal~\shortcite{Jalota2020MarketsFE} proposed a mechanism for allocating public goods that are capacity constrained due to social
distancing protocols, focusing on achieving a market clearing outcome.
Mehrotra \etal~\shortcite{Mehrotra2020AMF} studied the allocation of ventilators under a stochastic
optimization framework, minimizing the expected number of
shortages in ventilators while also considering the cost of
transporting ventilators.
Zenteno~\shortcite{zenteno2013models} combined influenza
modeling techniques with robust optimization to handle workforce shortfall in
a pandemic.
These papers differ from our work in that they do not explicitly address any fairness considerations that we study in this paper.

There has also been a growing interest in developing resource allocation mechanisms with fairness properties.
Under a fairly general class of utility functions including the Leontief
utility, computing market equilibrium under the fisher market setting
(divisible goods) can be done using an Eisenberg-Gale (EG) convex program 
Eisenberg and Gale~\shortcite{eisenberg1959consensus}. Market equilibrium solutions satisfy Pareto
optimality, proportionality, and envy-freeness. It is also known that an EG convex
program implicitly maximizes Nash welfare, which is the product of all
agents' utilities. However, MNW is generally
not strategy-proof, and
can be computationally expensive for large problems. 
Chen \etal~\shortcite{linearComboLtf} studied the computation and approximation of market equilibria under the so-called ``hybrid Linear-Leontief'' utilities. They assume that there are $m$ disjoint groups and each group contains several types of resources. An agent's utility is a linear combination of $m$ Leontief utilities, each associated with a group. Note that although we use the terms ``meta-type'' and ``group'' in our problem formulation, our setting is different from theirs because there is only one Leontief utility per user in our setting.

For Leontief utilities, Ghodsi \etal~\shortcite{DRF} introduced the Dominant Resource Fairness
allocation mechanism (DRF) which in addition to the three properties satisfied by
market equilibrium solutions, is also strategy-proof. Later
Parkes \etal~\shortcite{beyondDRF} extended the setting to allow agents to be weighted and have zero demand
over some resources while maintaining all four desiderata.
Our paper generalizes the setting further by allowing accessibility constraints, which as mentioned in the introduction, arise naturally in many practical settings.

For indivisible resources, Caragiannis \etal~\shortcite{unreasonable_fairness} showed that
maximizing Nash welfare with integer constraints (Discrete MNW) satisfies envy-freeness up
to one resource unit and has nice guarantees on the Max-Min Share ratio. However, similar to MNW, it is not strategy-proof and as we will see in the numerical section, it does not scale well to large number of agents and resource types. Although
exact market equilibrium might not exist in indivisible settings,
Budish~\shortcite{aceei_theory} showed that a close approximation of it exists in the
unweighted, binary allocation case. This was later put into practice for course
allocation in Budish \etal~\shortcite{course_match}. However, the theory does not provide useful
approximation bounds when assignments are not binary (e.g., each student only
needs one seat from each class, but each hospital may require hundreds of
doctors), and therefore it is not well suited to our setting.

\section{Problem Formulation}
\label{sec:problem_formulation}
For the remainder of the paper we use local medical personnel allocation as
a running example, even though other resource allocation problems can be
formulated in a similar fashion.
We group resources into \emph{meta-types}: doctors, nurses,
ventilators, emergency field hospital beds, etc. Within each meta-type (e.g.
doctors), we have \emph{types} (e.g. doctors from the Bronx, doctors from
Brooklyn, doctors from Manhattan, etc.\footnote{Manhattan, the Bronx, and Brooklyn are three boroughs of New York City.}). We assume that demands
are given over meta-types, but each agent sometimes can only receive
allocation from a subset of the resources in a meta-type because of
constraints such as location (e.g. a hospital is indifferent to where doctors
assigned to it come from as long as they are within commutable radius).
We refer to such subsets of resource types in each meta-type as the agents' demand \emph{groups}.

Let $\Omega_1, \Omega_2, \ldots, \Omega_L$ denote the
meta-types. Each meta-type $\Omega_l$ is a collection of resource types. We assume that $\Omega_i \cap \Omega_j = \emptyset$ for
any two different meta-types $i,j$, which means that each resource type belongs to exactly one meta-type. Let $R$ denote the set of all resource types: $R =
\cup_{l \in [L]}\Omega_l$, and $N$ denote
the set of agents. We use $m = |R|$ and $n = |N|$ to denote the total number of resource types and agents.
Each type of resource $r$ has a finite supply of $S_r$. We assume that the supplies are normalized within each meta-type: 
$$\sum_{r \in \Omega_l}S_r = 1\;\; \forall l \in [L].$$ 
Each agent $i \in
N$ submits a demand vector $[d_{i1}, \ldots, d_{iL}]$ where $d_{il}$ denotes the 
fraction of
meta-type $l$ that agent $i$ needs in order to get one unit of utility (one can think of this as each agent trying to  complete as many units of work as possible, where each unit of work requires $d_{il}$ units of meta-type $l$). 
Additionally, 
each agent has access to only a subset of resource types within each meta-type. We represent this accessibility constraint in the form of a set of demand groups.
Let $G_i =
\{g_l^{i} \subseteq \Omega_{l}: l \in [L], d_{il} > 0\}$, be the set of demand groups for agent $i$, where $g_l^{i} \subseteq \Omega_{l}$ is agent $i$'s demand group for $l$, specifying the subset of resource types belonging to meta-type $l$ that agent $i$ can access. 
Note that we only include in $G_i$ meta-types that agent $i$ has non-zero demand of. This is to simplify notation in the later analysis.
Intuitively, the introduction of meta-types models the substitution effects, and the introduction of demand groups models the accessibility constraints.
When $i$ is clear from the context, we sometimes use $g_l$ instead of $g_l^i$
to simplify the notation.

Following the setup in Parkes \etal~\shortcite{beyondDRF}, we also allow agents to be weighted differently for \emph{each} meta-type and we denote the weight of agent $i$ for meta-type $l$ as $w_{il}$. Having different weights for different meta-types makes the model more expressive: if we let $w_{i1}=\ldots=w_{iL}$, then this reduces to having a single priority weight for each agent.
This weight can depend on factors such as agent $i$'s contribution to the resource pool of meta-type $l$, as well as the size and stress level of agent $i$. 
In the case of medical supplies allocation, weights can represent how much each hospital is in need of extra resources. In the cloud compute setting, weights can represent how much money each user has paid for each meta-type of resource. Note that weights are fixed apriori, not self reported by the agents, nor determined by the allocation algorithm. 
We assume that weights are normalized within each meta-type: $\sum_{i\in N} w_{il} = 1$ for $l \in [L]$. 
Note that weights represent agents' priorities over the meta-types, not agents' preferences. Therefore they do not appear in the agents' utility functions, as we will define next. 

Let $x_{i}$ be the allocation vector of agent $i$: $x_{ir}$ represents the assignment of resource type $r$ to $i$. For each meta-type~$l$, $\sum_{r\in g^i_l} x_{ir}$ is the fraction of the total supply of meta-type $l$ that is assigned to agent $i$.   
The utility of agent $i$ is then defined as:
\begin{equation}
u_i(x_i) \coloneqq \min\limits_{g_l\in G_i}\left\{ \frac{1}{d_{il}} \sum\limits_{r\in g_l} x_{ir} 
 \right\}
\label{eq:leontief utility}
\end{equation}
Since agent $i$ needs $d_{il}$ fraction of each meta-type $l$ to finish one unit of work, $u_i(x_i)$ is the total units of work that agent $i$ can finish given
allocation vector $x_i$. This form of utility measure is called
the \emph{Leontief utility}. 

We now give a concrete example, which is also illustrated in Figure \ref{fig:example}. To keep the example simple we assume that each agent has the same weight over all meta-types: only one weight $w_i$ is defined for each agent $i$.
\begin{example}
\label{ex:example}
\begin{figure}
    \centering
    \includegraphics[width=0.7\linewidth]{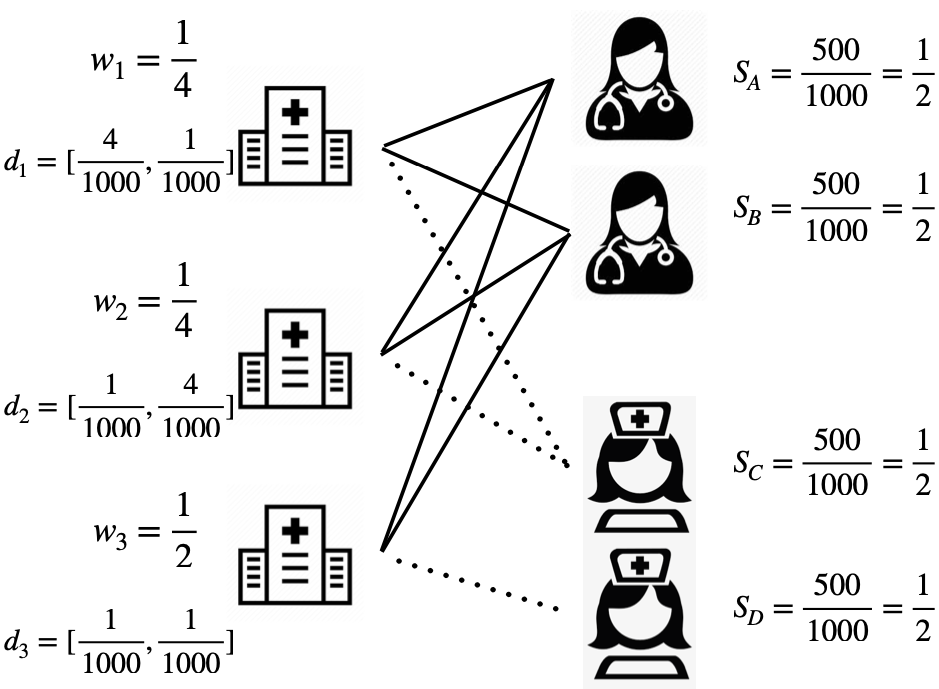}
    \caption{ All three hospitals can accept both types of doctors. However, hospitals I and II can only accept Nurse type C, while hospital III accepts only Nurse type D. }
    \label{fig:example}
\end{figure}

Consider a case of three hospitals (agents) $\{1, 2, 3\} $ and two resource meta-types. The first meta-type consists of two types of doctors (resource $A, B$), and the second consists of two types of nurses (resource $C, D$):
$\Omega_1 = \{A, B\}, \Omega_2=\{C, D\}$. The normalized weights for the three hospitals are: $w_1 = w_2 = \frac{1}{4}, w_3 = \frac{1}{2}$. The supply for each type of doctor and nurse is 500. Thus, the total available units of each meta-type is $500+500 = 1000$, and $S_r = \frac{500}{1000} = \frac{1}{2}\quad \forall r \in \{A, B, C,D\}$. All three hospitals have access to both types of doctors but hospitals $1,2$ only have access to nurse type $C$ while the third hospital only has access to nurse type $D$: $G_1=\{ g_1^1=\{A, B\}, g_2^1=\{C \} \}, G_2=\{ g_1^2=\{A, B\}, g_2^2=\{C \} \}, G_3 = \{g_1^3=\{A, B\}, g_2^3=\{D\}\}$. For each unit of utility, hospital 1 demands 4 doctors and 1 nurse, hospital 2 demands 1 doctor and 4 nurses, and hospital 3 demands 1 doctor and 1 nurse. Since the total units of supply for each meta-type is $1000$, 
$d_1 = [\frac{4}{1000},\frac{1}{1000}]$, $d_2 = [\frac{1}{1000},\frac{4}{1000}]$, $d_3 = [\frac{1}{1000}, \frac{1}{1000}]$.

\end{example}

\subsection{Desirable Properties}
\label{subsec:fairness definition}
Now we formally define the properties studied in this paper.
\paragraph{Pareto optimality.} 
An allocation mechanism is Pareto optimal if compared to the output allocation $x$, there does not exist another allocation $x'$ where some agent is strictly better off without some other agent being strictly worse off: $\exists i \;s.t.\; u_i(x'_i) > u_i(x_i) \implies \exists j \; s.t.\; u_j(x'_j) < u_j(x_j)$.

\paragraph{Weighted envy-freeness.}
Given an allocation $x_j$ for agent$j$, let $\tilde x_i$ be the same allocation adjusted to agent~$i$'s weights and demand groups, i.e., $\tilde x_{ir} = x_{jr}\frac{w_{il}}{w_{jl}}$ for all $r\in g_l^i, l\in[L]$, and $\tilde x_{ir} =0$ otherwise. 
$u_i\left(\tilde x_i\right) - u_i(x_i)$ is how much $i$ envies $j$. An allocation is weighted envy free if for any $i, j \in N$ this quantity is non-positive, i.e.,
    $$u_i\left(\tilde x_i\right) - u_i(x_i) \leq 0.$$
Intuitively, this means an agent prefers her allocation over the allocation of any other agent scaled by the weight ratios of the two agents. Note that since there is a separate weight for every meta-type $l$, the allocations for each resource type $r$ is scaled according to the corresponding weight for the meta-type that it belongs to.

\paragraph{Strategy-proofness.} In the existing literature, agents can only be
strategic by misreporting their demand vector. In our setting however,
agents have the additional possibility of misreporting their
accessibility constraints for the meta-types
(e.g. One can report that she accepts both Intel
and AMD CPUs but in fact her program only runs on Intel CPU). Our definition
of strategy-proofness guards against both types of misreporting.
Let $x$ be the allocation returned by the
mechanism under truthful reporting from all agents. Let $x'$ be an allocation
returned by the mechanism when all agents report truthfully except agent $i$
reports an alternative demand vector and/or alternative demand groups.
The mechanism is strategy-proof if $u_i(x_i) \geq u_i(x'_i)$.

\paragraph{Sharing incentive.}
In settings where the supplies for each resource come from the participants' contribution, sharing incentive is satisfied when the resulting allocation gives each participant at least as much utility as she originally had. 
More specifically, for each $i\in N$ and $l\in [L]$, let $s_{il}$ be the proportion of meta-type $l$ contributed by agent $i$ \emph{that she can also access}. We can also think of $s_{il}$ as the amount of ``useful'' resource agent $i$ originally possessed of meta-type $l$ (she might contribute more than $s_{il}$ to the pool). Prior to reallocation of resources, agent $i$'s utility would be
$$u_i^o := \min_{g_l \in G_i} \left \{\frac{s_{il}}{d_{il}} \right \}.$$
Sharing incentive says that $u_i(x_i) \geq u_i^o \;\; \forall i \in N$, where $x_i$ is the output allocation of the algorithm (i.e., all agents have incentives to share (pool) their individual resources for reallocation).\footnote{A closely related concept called proportionality is also often seen in the literature. We focus on Sharing Incentive in the main paper and include a discussion of proportionality in the Appendix.} 



\section{Dominant Resource Fairness with Meta-Types}
\label{algorithm}
Before describing the algorithm we first define some key concepts used in the algorithm:
\begin{equation*}
l_i^* \coloneqq \arg\min\limits_{l \in [L]} \frac{w_{il}}{d_{il}} \;\;\;\; d_{i*} \coloneqq d_{il^*_i} \;\;\;\; w_{i*} \coloneqq w_{il^*_i}
\end{equation*}
Namely, $l_i^*$ is the meta-type from which agent $i$
demands the biggest proportional share, adjusted by her priority weights. We refer to $l_i^*$ as the
\emph{dominant resource
meta-type} for agent~$i$. $d_{i*}$ is the proportional share demanded by agent $i$ from its dominant resource
meta-type to finish one unit of work.

We now present our fair allocation mechanism, which we call \emph{Dominant Resource Fairness with Meta-Types} (DRF-MT).
The mechanism proceeds in rounds and agents are gradually ``eliminated''. In each round $t$, we use the linear program described in \eqref{opt prob} to maximize a fractional value $y_t$ so that each \emph{remaining agent} $i$ ($i\in N_t$) receives at least $y_t \times w_{i*}$ fraction of the total supply from its' dominant resource meta-type $l_i^*$, and more generally $y_t\times w_{i*}\times d_{il}/d_{i*}$ of each demanded meta-type $l$.\footnote{We also considered an alternative design that did not yield a strategy-proof and envy-free algorithm. See Appendix~D
.}
Based on this solution, 
we eliminate at least one resource and one agent using Definition~\ref{eliminated resource} and ~\ref{eliminated agent} (although the algorithm only needs to explicitly maintain a list of active/eliminated agents, not resources).
For each agent $i$ eliminated in round $t$, we set $\gamma_i = y_t$. We fix the fraction of dominant meta-type $l_i^*$ assigned to agent $i$ to  $\gamma_i \times w_{i*}$, \emph{without fixing the specific allocations of the
resources}. We first observe the following fact (proof is in the Appendix):
\begin{fact}
\label{fact:tightallocation}
    In any round $t$ of Algorithm~\ref{extendedDRF}, the allocation constraints in Equation \ref{opt prob} for $i\notin N_t$ are tight for optimal solutions.
\end{fact}
This fact implies that when an agent is eliminated, her utility in the final allocation is fixed, even though the exact allocation is not. Not fixing the allocation is a deliberate algorithmic design choice because agents who are flexible with their demand groups should accommodate agents who are more restrictive (e.g. if agent~1 accepts both type A and B, and agent~2 only accepts type A, then we should allocate agent 1 mostly type B resource, and leave type A resource for agent~2). When the number of agents and resource types is large, it is difficult to characterize such dynamics explicitly. So it is crucial to not fix the allocation to the agents until the last iteration.


We will show that there is at least one new
resource and one agent being eliminated in each round. Thus our algorithm requires at most $\min(m, n)$ rounds (in practice it often terminates in 2-3 rounds even with a large number of resources types and agents). Since each round requires solving a polynomial-sized linear program, the overall procedure can be run in polynomial time.

Let $N_t, R_t$ be the set of active agents and resources at the
beginning of round $t$. The LP for round $t$ is defined in \eqref{opt prob}. Note that the ratio $\frac{d_{il}}{d_{i*}}$ is simply making sure that there is no waste in the allocation. For an agent who has been eliminated, $\frac{\gamma_iw_{i*}}{d_{i*}}$ is her final utility. If agent $i$ is not yet eliminated after round $t$, then $\frac{y_tw_{i*}}{d_{i*}}$ represents how much utility she is currently guaranteed to receive (it will never decrease in later rounds, see Fact \ref{fact:monotonicity of lp over rounds}).
\begin{align}
    &\max \; y_t &\nonumber\\
    \text{s.t. }  & \text{(active agents allocation constraints)}\nonumber\\
    & y_t \times w_{i*}\times\frac{d_{il}}{d_{i*}} \leq \sum\limits_{r\in g_l} x_{ir} \quad \forall i\in N_t, g_l\in G_i \nonumber\\
    & \text{(eliminated agents allocation constraints)}\nonumber\\
    &\gamma_i \times w_{i*}\times \frac{d_{il}}{d_{i*}} \leq \sum\limits_{r\in g_l} x_{ir} \quad \forall i \not\in N_t, g_l\in G_i &\label{opt prob}\\
    & \text{(supply constraints)}\nonumber\\
    &\sum\limits_{i\in N} x_{ir} \leq S_r \quad\forall r\in R \nonumber\\
    & \text{(non-negativity constraints)}\nonumber\\
    &x_{ir} \geq 0 \quad \forall i\in N, r\in R\nonumber
\end{align}

\begin{fact}
    \label{fact:monotonicity of lp over rounds}
    The optimal value for Equation \ref{opt prob} is non-decreasing over rounds: $y^*_1\leq y^*_2 \leq ...$, where $y^*_t$ is the optimal objective function value of the LP in round $t$.
\end{fact}
This follows because the
constraints on eliminated agents are less restrictive than the constraints on
active agents, and the set of active agents is decreasing over time. 

\begin{definition}
    \label{eliminated resource}
    Resource $r$ is eliminated in round $t$ if $t$ is the first round in Algorithm~\ref{extendedDRF} in which $\sum_{i\in N} x_{ir} = S_r$ for
    every optimal $x$.
\end{definition}
By Fact~\ref{fact:monotonicity of lp over rounds} it is also easy to see that the set of remaining resources $R_t$ decreases over time.
\begin{definition}
    \label{eliminated agent}
    We give two equivalent definitions for eliminating agents:
    \begin{itemize}
        \item Agent $i$ is eliminated in round $t$ when there exists $g_l\in G_i$ such that $g_l\cap R_{t+1}=\emptyset $.
        \item Agent $i$ is eliminated in round $t$ when there exists $g_l\in G_i$ such that $y_t \times w_{i*}\times\frac{d_{il}}{d_{i*}} = \sum\limits_{r\in g_l} x_{ir}$ for every optimal $x, y_t$.
    \end{itemize}
    
\end{definition}
Intuitively, both definitions are saying that agent $i$ can not improve her utility further in later rounds. 
Due to space constraint we defer the proof of their equivalence, and most of the other results to the Appendix. We include the proof of Claim~\ref{at least one resource elimination} here because it is a good representation of the flavor of arguments used in other proofs. First we address the question of whether the DRF-MT can be efficiently implemented.

\begin{algorithm}
    \SetAlgoLined
    Input: Agents $N$, resources $R$, supplies $S_r$ $\forall r \in R$, demand groups $G_i \; \forall i \in N$, normalized demands $d_{il} \;\forall i\in N, g_l\in G_i$, priority weights $w_{il}\; \forall i\in N, l \in [L]$\\
    Initialize $N_0 = N$\\
    \For{$t\gets 0,1,2,...$}{
        $y_t^* \gets $ Solve \eqref{opt prob}\\ 
        Update the remaining active agents $N_{t+1}$ (using Claim \ref{efficient computation})\\
        \For{agent $i$ eliminated in this round}{
            $\gamma_i \gets y_t^* $
        }
        \If{$N_{t+1} = \emptyset$}{
            Solve Equation \ref{opt prob} and assign resources according to $x_{ir}$ with rounding\\
            break
        }
    }
    \caption{Dominant Resource Fairness with Meta-Types (DRF-MT) }
    \label{extendedDRF}
\end{algorithm}

\begin{claim}
In each round $t$ of Algorithm~\ref{extendedDRF}, at least one remaining resource $r\in R_t$ and one remaining agent $i\in N_t$ is eliminated. 
\label{at least one resource elimination}
\end{claim}

\begin{proof}
    Suppose no resource is eliminated in round $t$, then for each $r\in R_t$, there exists an optimal solution
    such that $\sum_{i\in N} x_{ir} < S_r$. Then the convex
    combination of these solutions gives us an optimal solution $x^*$ that
    satisfies $\sum_{i\in N} x^*_{i, r} < S_r \;\forall r\in R_t$.
    However, by Definition \ref{eliminated agent}, for every remaining agent $i\in N_t$,
    $g_l \cap R_t \neq \emptyset\,\forall g_l\in G_i$. So if we assign
    a little more of every active resource to every active agent, then the
    overall objective value would be higher. This contradicts the optimality of the LP.
    
    Now suppose some resource $r\in R_t$ is eliminated in round $t$ but no agent is eliminated. Suppose this resource type is part of meta-type $l$. By the first definition in Definition \ref{eliminated agent}, this means that for every $i$ such that $x_{ir}>0$, there exists $ r'\in g^i_l$ such that $r'$ is not eliminated. By the same convex combination argument above, we know that there is an optimal solution such that $\sum_i x_{ir'} < S_{r'}$ for every such $r'$. Then for every such agent we can remove $\epsilon $ allocation of $r$ from her and replace it with $\epsilon$ allocation of the corresponding $r'$. This gives us an allocation that has the same objective as before without using up the entire supply of $r$, contradicting $r$ being eliminated. 
\end{proof}
This result shows that DRF-MT can be implemented
efficiently by solving at most $\min(m, n)$ number of polynomial-size linear programs. However, it does not tell us how to find the eliminated agents. The following theorem says that we can do so by looking at the dual variables of the LP. Note that the algorithm does not need to explicitly maintain a list of active resources (Equation~\ref{opt prob} does not depend on $R_t$).
\begin{claim}
This claim has two parts:
\begin{itemize}
\item If the shadow price of an allocation constraint of an active agent in round $t$ is positive, then its corresponding agent is eliminated in round $t$.
\item In each round $t$, at least one allocation constraint corresponding to an agent in $N_t$ has a positive shadow price.
\end{itemize}
\label{efficient computation}
\end{claim}

Now we state our main results.

\begin{lemma}
    DRF-MT is Pareto optimal.
    \label{lemma: PO}
\end{lemma}

\begin{lemma}
    DRF-MT is weighted envy-free.
    \label{lemma: EF}
\end{lemma}

\begin{lemma}
    DRF-MT is strategy-proof.
    \label{lemma:SP}
\end{lemma}
The proofs for these three lemmas all involve a case analysis of different scenarios and showing that the undesirable outcomes violate either the optimality of the LP or the definition of eliminated resources/agents, similar to the arguments presented in the proof of Claim~\ref{at least one resource elimination}.

\begin{lemma}
    Assume that demands, weights and supplies are all rational numbers. If priority weights of the algorithm are set according to the each agent's accessible contribution to the resource pool (for each meta-type), then DRF-MT satisfies sharing incentive.
    \label{lemma: sharing incentive}
\end{lemma}
In resource pooling settings, having a separate weight for each meta-type is crucial in proving this result (e.g. contributing a ton of hard drive space but no GPU should not give the agent high priority if GPU is its dominant/bottleneck resource type). The proof constructs a bipartite graph of supplies and demands of the resources, then uses Hall's Theorem Hall~\shortcite{hallstheorem1935} to show that there exists a feasible solution to the first round's LP that already gives every agent at least as much utility as they could get without participating in the pool. Since agents' utilities only increase in later rounds, the final allocation must also satisfy sharing incentive. 
\subsection{Integral Allocation from Rounding}
\label{sec: rounding}
So far we have implicitly assumed that the resources are
divisible, and all fairness results are stated with respect to the \emph{fractional
assignment} output of Algorithm~\ref{extendedDRF}. In practice we round down the
output to obtain the final assignment, since resources
such as ventilators are indivisible.
Each agent loses at most 1 unit of each type of resource through rounding. Since we focus on problems where each agent receives hundreds of units of each resource, 
the performance loss due to rounding is small. For example, starting with an envy-free fractional allocation, one agent can envy another by at most $2m$ items after rounding. In Example \ref{ex:example}, $m$ is $4$, while the total allocations each agent receives are in the hundreds. So an envy of $2m$ items is not significant.
Note that such divisibility assumption is also standard in existing DRF literature, which often focus on the compute resource sharing problem:
even though CPU cores are discrete, it's common to treat the problem as a continuous problem since there is a large quantity of cores in a compute cluster. 

There are many existing algorithms that focus on fair allocation of indivisible goods (e.g. Discrete MNW from
Caragiannis \etal~\shortcite{unreasonable_fairness}). Indivisible resource allocation is particularly
important when the quantities of the resources are small (e.g.
fairly assigning a car, a house, and a ring to two people). However, as is the case with most discrete optimization problems, these
algorithms do not scale well to the sizes that we deal with in a pandemic with hundreds of hospitals and 
many types of resources.
In settings where each agent receives hundreds of units of each
resource, the performance loss due to rounding is often small compared to the
dramatic increase in computational cost for solving Mixed Integer Programs (see Section~\ref{experiments} for a numerical comparison).

\subsection{Connection to Previous Dominant Resource Fairness Algorithms}
The core difference between our problem setup and the existing DRF settings (Ghodsi \etal~\shortcite{DRF}, Parkes \etal~\shortcite{beyondDRF}) is the addition of accessibility constraints. 
When $|\Omega_l| = 1$ for each $l \in [L]$, both our problem formulation and the DRF-MT algorithm reduce to the problem and algorithm studied in those papers. 
Note that in this simplified setting one can write out the closed form solutions to the LPs, so no actual optimization needs to be performed. However, it is natural that resources come in different ``flavors'' and that agents have different constraints/preferences over these variations. So our formulation captures a much wider range of problems encountered in practice.



\subsection{Alternative Fair Allocation Mechanisms}
As discussed in Section~\ref{sec:relatedwork} and Section~\ref{sec: rounding},
the other most suitable approaches in our setting are MNW and Discrete MNW. When
the weights are equal, MNW is also commonly referred to as the Competitive
Equilibrium with Equal Income (CEEI) approach. Without the accessibility
constraints, MNW is known to be Pareto optimal, envy-free and satisfy sharing
incentive. However, unlike DRF-MT, it is known that MNW is not strategy-proof
(see Section 6 of Ghodsi \etal~\shortcite{DRF} for an example).
We show in Section~\ref{experiments} that our DRF-MT mechanism achieves almost
as much social welfare (i.e. sum of utilities of all agents) as MNW, and also
runs faster in practice.

\begin{figure*}[]
  \centering
  \begin{minipage}{.33\textwidth}
    \includegraphics[width=1.0\linewidth]{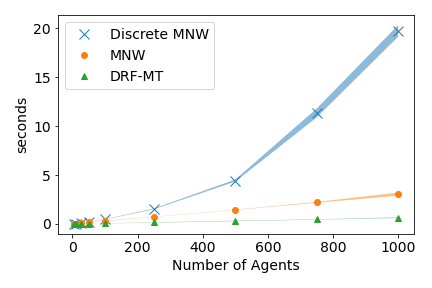}
  \end{minipage}%
  \begin{minipage}{.33\textwidth }
    \includegraphics[width=1.0\linewidth]{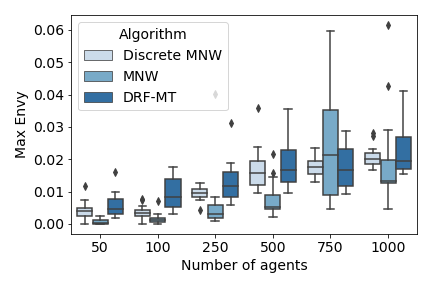}
  \end{minipage}%
  \begin{minipage}{.33\textwidth }
    \includegraphics[width=1.0\linewidth]{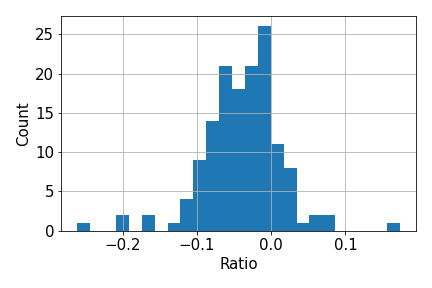}
  \end{minipage}%
  \caption{ Left: Running time comparison. Middle: Normalized max envy
    comparison. Right: Distribution of normalized difference in social welfare
    between Discrete MNW and DRF-MT over all trials. Normalized difference is
    calculated by subtracting the social welfare of Discrete MNW from that of
    (rounded) DRF-MT and then dividing by the social welfare of Discrete MNW. }
  \label{fig:main experiments}
\end{figure*}

\subsection{Extension to Arbitrary Group Structure}\label{extended-work}
We currently assume that resources and demands follow a meta-type/group/type structure. One might be interested in a general group structure where a
demand group can contain any subset of all possible resources (not necessarily
from a single meta-type). The problem with this kind of flexible group structure is
that it opens up possibilities for people to cheat the system by misreporting
their true demand structure (e.g. instead of reporting that they are indifferent
to resource A and B, and that they only need one unit of either one to finish a
unit of work, agents can claim that they need one unit each from both A and B to finish one unit of work).
In particular, Dominant Resource Fairness based approaches will likely not work, since it is unclear how one would even define the dominant resource under such a general setting. We leave this as an open question for future work.

\section{Numerical Experiments}
\label{experiments}
We compare the algorithms on running time, normalized max envy, and social welfare. Normalized max envy is the maximum envy (see Section~\ref{subsec:fairness definition}) between any pair of agents normalized as a fraction of each agent's allocated utility. Social welfare is the sum of utilities of all agents. We fix a meta-type structure ($\Omega_1 = \{0\}, \Omega_2=\{1, 2\}, \Omega_3=\{3, 4, 5\}, \Omega_4=\{6, 7, 8, 9\}$) and randomly generate the demands, group structures, and weights for the agents. For each choice of number of agents, we ran 16 trials. All three algorithms Allocations are rounded down for MNW and DRF-MT. All three algorithms allow specifying different agent weights and also observe the accessibility constraints. Gurobi~\shortcite{gurobi} and Mosek (Aps~\shortcite{aps2020mosek}) are used to implement the algorithms. More details on the experimental setup and additional experiments can be found in 
Appendix~B
.

First we investigate scalability. 
As shown in Figure~\ref{fig:main experiments} (left), the running time for Discrete MNW quickly explodes while MNW and DRF-MT are much more scalable. DRF-MT runtime in particular grows very slowly.
The error region represents one standard deviation from the mean. 

Recall that DRF-MT is envy-free before rounding. We now investigate envy when the solution is rounded.
Without accessibility constraints, MNW is also envy-free before rounding, and Discrete MNW satisfies envy-free up to one good. Figure~\ref{fig:main experiments} (middle) shows that all three algorithms have small max envy after rounding in practice ($< 4\% $ in most trials).

Finally, we compare the social welfare obtained under DRT-MT and Discrete MNW. 
Figure~\ref{fig:main experiments} (right) shows that in roughly 95\% of the trials, DRF-MT obtained at least 90\% as much social welfare as Discrete MNW. 
In Appendix~B
we show that MNW has slightly lower social welfare than Discrete MNW, so the above conclusion holds when DRF-MT is compared to rounded MNW as well.

In conclusion, compared to both Discrete MNW and MNW,  DRF-MT 1) achieves almost as much social welfare, 2) has comparable level of max envy, 3) has the additional property of strategy-proofness (in the fractional case), and 4) is more scalable. An interesting avenue for future work is to determine the properties of the rounded variant of DRF-MT. A particularly interesting question would be whether one can show that approximate strategy-proofness holds when there is large supply of each item.

\clearpage
\bibliographystyle{named}
\bibliography{ijcai21}

\clearpage

\appendix

\section{Missing Proofs of Results}
\label{appendix}
\subsection{Proof of Equivalence in Definition~\ref{eliminated agent}}
\begin{proof}
    Suppose the second definition holds but the first one does not. Then by the definition of eliminated resource, there exists an optimal solution such that for every $g_l\in G_i$, there exists $r\in g_l$ such that $\sum_{i\in N} x_{ir} < S_r$. Then for every $g_l\in G_i$ we can assign $i$ a little more of the resource type above, and have $y_t\times w_{i*} \times \frac{d_{il}}{d_{i*}} < \sum_{r\in g_l} x_{ir}$. This contradicts the second definition.
    
    Now suppose the first definition holds but the second definition does not. This means that there exists an optimal solution such that $y_t\times w_{i*} \times \frac{d_{il}}{d_{i*}} < \sum_{r\in g_l} x_{ir}$ for every $g_l\in G_i$. Consider the $g_l$ such that $g_l\cap R_{t+1}=\emptyset$ by the first definition (every $r\in g_l$ is eliminated by the end of round $t$). We can reduce the allocation of resources in that demand group to $i$ by a little bit without sacrificing optimality because the allocation constraints were satisfied strictly. But this means we have an optimal solution that does not use up the supply of $r\in g_l$: this contradicts the elimination of these resources in the first definition. 
\end{proof}

\subsection{Proof of Claim~\ref{efficient computation}}
\begin{proof}
    The first part is straightforward. 
    If $q_{ig}>0$ is the dual variable for the allocation constraint for some agent $i\in N_t, g\in G_i$, then by complementary slackness every optimal solution needs to satisfy $y_t \times w_{i*}\times\frac{d_{il}}{d_{i*}} = \sum\limits_{r\in g_l} x_{ir}$, which means agent $i$ needs to be eliminated by Definition~\ref{eliminated agent}.

    Let's now rewrite the linear program solved at time $t$:
    \begin{align}
        &\max  y_t &&&&\label{opt prob reformulated}\\
        \text{s.t.}  
                    -&\sum\limits_{r\in g_l} x_{ir} + &\frac{d_{il}}{d_{i*}} w_i y_t \leq& 0  &\forall& i\in N_t, g_l\in G_i \label{active allocation constraints}\\
                    -&\sum\limits_{r\in g_l} x_{ir} &\leq& - \frac{d_{il}}{d_{i*}} \gamma_i w_i  &\forall& i \not\in N_t,
                        g_l\in G_i \label{eliminated allocation constraints}\\
                     &\sum\limits_{i\in N} x_{ir} 
                        &\leq& S_r &\forall& r\in R \label{supply constraints}\\
                     &x_{ir} &\geq& 0 &\forall& i\in N, r\in R\nonumber
    \end{align}

    This LP is in canonical form, where the objective coefficient vector is $c^T
    = [0,...,0,1]$. Let $q_{ig}$ be the dual variables that correspond to the
    allocation constraints (constraint \eqref{active allocation constraints} and \eqref{eliminated allocation constraints}), and $q_r$ the dual variables corresponding to the
    supply constraints (constraint \eqref{supply constraints}). Let $y^*_t$ be the value of $y_t$ in an optimal solution
    to the linear program and let $\bar q$ be the optimal solution to the
    corresponding dual program.
    By complementary slackness we know that $\bar q ^\top A_{y} =
    c_{y} = 1$, where $A_y$ is the last column of the primal constraint
    matrix. Note that the entries in $A_y$ are either positive or zero.
    Therefore, $\bar q_{ig}$ must be positive for some $i\in N_t, g\in G_i$. This finishes the proof of the second part.
    
\end{proof}

\subsection{Proof of Lemma~\ref{lemma: PO} and Fact~\ref{fact:tightallocation}}

\begin{proof}
    Suppose $x$ is the output of Algorithm \ref{extendedDRF} and there exists
    allocation $x'$ such that agent $i$ is strictly better off while other
    agents have just as much utility. Let $y'\times w_{i*}$ be the fraction of the dominant resource meta-type that $i$ receives with allocation $x'$.

    Let $t$ be the round in which $i$ was eliminated in Algorithm \ref{extendedDRF}. 
    Since $i$ is strictly better off with allocation $x'$, 
    $y'>y^*_t$. Now we construct a new allocation by scaling down
    agent $i$'s allocation from $x'_i$ to $x'_i \frac{y^*_t}{y'}$. 
    Since we know other agents have at least as much utility as with allocation $x$, this new solution has an LP objective value at least as high as before, satisfies all the allocation/supply constraints, and does \emph{not} use up all the resources that $i$ cares about. This contradicts $i$ being
    eliminated in round $t$. This concludes the proof for Lemma~\ref{lemma: PO}

    Note that by the same argument as above we know that the allocation constraint in Equation \ref{opt prob} for the eliminated agents has to be satisfied with equality (otherwise we can scale this allocation down to make the constraint tight, and that agent would not have been eliminated in an earlier round).
\end{proof}

\subsection{Proof of Lemma~\ref{lemma: EF}}
\begin{proof}
    For any pair of agents $i, j \in N$, we will show that $i$ does not envy $j$. Let $x$ be the allocation returned by Algorithm \ref{extendedDRF}. Starting from the LHS of the definition of weighted envy-freeness:
    \begin{align*}
        & u_i\left(x_{jr}\frac{w_{il}}{w_{jl}}\,\, \forall r\in g_l^i, l\in[L]\right)\\
        &= \min\limits_{g_l \in G_i} \frac{\sum\limits_{r\in g^i_l}x_{jr}\frac{w_{il}}{w_{jl}}}{d_{il}}\\ 
        &= \min\limits_{g_l \in G_i} \frac{\sum\limits_{r\in g^j_l \cap g^i_l}x_{jr}\frac{w_{il}}{w_{jl}}}{d_{il}}\\ 
        &\leq \min\limits_{g_l \in G_i} \frac{\sum\limits_{r\in g^j_l}x_{jr}\frac{w_{il}}{w_{jl}}}{d_{il}}\\
        & = \min\limits_{g_l \in G_i} \frac{\frac{w_{il}}{w_{jl}}\sum\limits_{r\in g^j_l}x_{jr}}{d_{il}}.
    \end{align*}
    
    The first equality is the definition of Leontief utility in \eqref{eq:leontief utility}. The second equality holds because $x_{jr} = 0$ for $r \in \Omega_l$ but $r\notin g_l^j$ (If the output allocation does contain inaccessible resources then we can simply remove them without affecting the utilities of agents). The inequality follows from non-negativity of $x_{jr}$.
    
    Now let $t_i$, $t_j$ be the rounds in which
    agent $i$ and $j$ are eliminated respectively.
    Note that from the LP in Equation \ref{opt prob}, we know $\sum\limits_{r\in g^j_l}x_{jr} = y^*_{t_j} w_{j*} d_{jl}/ d_{j*}$. So
    \begin{align*} 
        \min\limits_{g_l \in G_i} \frac{\frac{w_{il}}{w_{jl}}\sum\limits_{r\in g^j_l}x_{jr}}{d_{il}} &= \min\limits_{g_l \in G_i} \frac{\frac{w_{il}}{w_{jl}} y^*_{t_j} w_{j*} d_{jl}/ d_{j*}}{d_{il}}\\ 
        &= \min\limits_{g_l \in G_i} \frac{w_{j*}}{d_{j*}} \frac{d_{jl}}{w_{jl}}y^*_{t_j}\frac{w_{il}}{d_{il}}\\
        &\leq \min\limits_{g_l \in G_i} y^*_{t_j}\frac{w_{il}}{d_{il}} = y^*_{t_j}\frac{w_{i*}}{d_{i*}}
    \end{align*}
    
    where the inequality follows from the definition of dominant resource meta-type ($\frac{w_{j*}}{d_*^j} = \min\limits_{l} \frac{w_{jl}}{d_{jl}}$). 
    
    If $y^*_{t_j} \leq y^*_{t_i}$ (which means $t_j \leq t_i$, by Fact \ref{fact:monotonicity of lp over rounds} and Fact \ref{fact:tightallocation}), we have
    \begin{align*}
        y^*_{t_j}\frac{w_{i*}}{d_{i*}} \leq y^*_{t_i}\frac{w_{i*}}{d_{i*}} 
        =u_i(x_i).
    \end{align*}
    
    Now suppose $y^*_{t_j} > y^*_{t_i}$ (which means $t_j>t_i$), and that $i$ envies $j$. 
    Note that this implies that for every group $g_l^i \in G_i$, there exists at least one $r \in g_l^i$ such that $x_{jr}>0$. 
    
    Consider an alternative allocation $x'$
    that scales the allocation to agent $j$ to 
    $\frac{y^*_{t_i}}{y^*_{t_j}} x_j$ while keeping the allocations to other agents the same as in $x$, namely, $x'_j = x_j\frac{y^*_{t_i}}{y^*_{t_j}}$ and $x'_k = x_k \; \forall k\neq j$. This alternative allocation gives every agent as much utility as they had before in round $t_i$ while maintaining 
    slack in at
    least one resource from each demand group of $G_i$. This contradicts the definition of $t_i$ because agent $i$ was eliminated in round $t_i$ (see
    Definition \ref{eliminated agent}).
\end{proof}

\subsection{Proof of Lemma~\ref{lemma:SP}}
Our proof approach is adapted from \cite{beyondDRF} with important modifications. We first introduce some new notations and prove
two helpful results. 
Let $i$ be the only agent who reports her demands untruthfully. Let $d$ be the true demand vector for all agents
and $d'$ be an alternative demand 
where only the elements belonging to agent $i$
might be different. Let $t^*$ be the first round in which agent
$i$ is eliminated in Algorithm \ref{extendedDRF}, either with truthful or untruthful reporting (minimum of
the two). Let $N_t, N_t'$, and $y^*_t, y^{*\prime}_t$  represent the remaining active agents at the beginning of round $t$, and the optimal LP objective in round $t$, under $d$ and $d'$ respectively, 
\begin{claim}
    If agent $i$ is not eliminated in round $t$, then if we remove the
    allocation constraint for agent $i$ and omit the variables related to agent
    $i$ from the supply constraints in Equation \ref{opt prob}, the optimal
    value as well as agents eliminated in that round do not change.
    \label{free constraints of non-eliminated agents}
\end{claim}
\begin{proof}
    First we show that $x_{ir} = 0$ if $r$ is one of the eliminated
    resources in that round. Suppose $x_{ir}>0$. Since $i$ is not
    eliminated, there must exist another resource $r'$ in the same demand group of
    $r$ for agent $i$ that is not eliminated. This means that we could
    replace some of the allocation of $r$ with a little more
    allocation of $r'$. But this would then contradict $r$ being an
    eliminated resource. Note that by the same logic $x_{ir}=0$ holds in
    all future rounds too.
    
    This allows us to remove $x_{ir}$ from the supply constraints. Now the
    allocation constraint can be written as
    $$ 
    y_t \times w_{i*}\times\frac{d_{il}}{d_{i*}} \leq \sum\limits_{r\in  g_l\cap R_{t+1}} x_{ir} \quad \forall g_l\in G_i\nonumber\\
    $$

    Since the remaining resources are not constrained by supply, this
    inequality can always hold without posing limits on other variables. So we
    can remove this constraint completely.
\end{proof}

\begin{claim}
    For all $t\leq t^*$, $N_t = N_t'$. 
    For all $t< t^*$, $y^*_t = y_t^{*\prime}$. 
    \label{claim:same solution up to t star}
\end{claim}
\begin{proof}
    We use proof by induction. $t=0$ holds trivially.

    We assume the claim holds for $t$. Suppose $t+1 < t^*$. Then by Claim \ref{free constraints of non-eliminated agents}, we can remove the constraints
    related to agent $i$ from the optimization problem. But the only
    differences between these two optimization problems are those related to
    agent $i$, so they have the same solutions and we are eliminating the
    same agents.

\end{proof}

Now we prove Lemma \ref{lemma:SP}.
\begin{proof}

    Let $x$ and $x'$ be the allocations returned by Algorithm \ref{extendedDRF} given demand $d$ (truthful reporting) and $d'$ (agent $i$ misreports) respectively. We consider the following four cases separately.
    
    \begin{itemize}
        \item $y^*_{t^*}\leq y_{t^*}^{*\prime}$ and agent $i$ is eliminated in $t^*$
        reporting $d$. By Claim \ref{claim:same solution up to t star}, we know
        $N_{t^*}=N_{t^*}'$. Although we do not know the exact round in which
        agents in $N_{t^*}'$ are eventually eliminated under $d'$, we know that their
        dominant resource shares are all at least $y^{*\prime}_{t^*}\geq y^*_{t^*}$,
        because the optimal objective value of the optimization problem can
        only increase over time by Fact \ref{fact:monotonicity of lp over rounds}.

        Suppose $u_i(x_i') > u_i(x_i)$. Now
        consider $x'$ as a candidate solution for the optimization problem
        in round $t^*$ of truthful reporting. 
        Every agent $j$ in $N_{t^*}$ receives at least $y_{t^*}^{*\prime} w_{j*} \geq y^*_{t^*}w_{j*}$ fraction of their dominant resource meta-type, while agent $i$
        receives strictly more. This contradicts either the optimality of $y_t^*$ or the fact that agent $i$ was eliminated in round $t^*$ reporting $d$ (see Definition \ref{eliminated agent}).

        \item $y_{t^*}^*\geq y_{t^*}^{*\prime}$ and the agent is eliminated in $t^*$
        reporting $d'$. 
        Suppose the dominant resource meta-type is the same under $d'$. Since $y_{t^*}^{*\prime}\times w_{i*}$ is the fraction of the total supply of dominant resource that $i$ receives, agent $i$ must be receiving
        less of that under $d'$.
        
        Now suppose the reported dominant resource
        meta-type is different under $d'$. Let $\bullet$ denote the new dominant resource meta-type. Let $w_{i \bullet}$, $d_{i\bullet}^{\prime}$ be the new dominant resource weight and demand. Let $d_{i*}^{\prime}$ be the \emph{new} demand for the \emph{original} dominant resource meta-type. The amount of original dominant resource meta-type $i$ receives is
        \begin{align*}
            y_{t^*}^{*\prime} \frac{w_{i\bullet}}{d_{i\bullet}^{\prime}} d_{i*}^{\prime}&\leq y_{t^*}^{*\prime} w_{i*} \leq y_{t^*}^{*} w_{i*}
        \end{align*}
        The first inequality follows from the definition dominant resource meta-type: $\frac{w^i_\bullet}{d^{i\prime}_\bullet} := \min\limits_{l\in[L]} \frac{w_{il}}{d^{i\prime}_l}$.
        The final expression is the amount of \emph{original} dominant resource that agent $i$ receives under truthful reporting.

        \item $y_{t^*}^{*\prime} > y_{t^*}^*$ and the agent is not eliminated reporting
        $d$ but is eliminated reporting $d'$. We argue that this case cannot
        happen. By Claim \ref{claim:same solution up to t star}, we can remove the
        allocation constraints related to $i$ in round $t^*$ under truthful reporting. But
        then we are left with two optimization problems with the same
        constraints, except that with untruthful reporting the optimization
        problem has extra allocation constraint (for agent $i$), and an extra
        non-negative term in the supply constraints. Extra constraints
        and extra terms in the supply constraints can only make the
        optimization problem harder.

        \item $y_{t^*}^{*\prime} < y_{t^*}^*$ and the agent is eliminated reporting $d$
        but not eliminated reporting $d'$. This is the symmetric case as the
        previous one and so cannot happen either.

        Finally, a closer inspection of the above shows we did not need the group structure of agent $i$ to stay the same, so the result holds for misreporting group structures as well. 
    \end{itemize}    

\end{proof}

\subsection{Proof of Lemma~\ref{lemma: sharing incentive}}
\begin{proof}
    
    Recall that we use $s_{il}$ denote the proportion that is both accessible to and contributed by agent $i$. Each agent might have access to other people's contributions as well. We set $w_{il} = s_{il}$.
    Since an agent might not have access to all of the supplies that she brings, $\sum\limits_{i\in N}w_{il}$ might be strictly less than one. In that case we can pretend that there is a phantom agent with weight $1-\sum\limits_{i\in N}w_{il}$ for each meta-type $l$, and that his demand vector is zero. Note that we do not need to implement this phantom agent when running the algorithm, because DRF-MT is invariant to the scale of weights. We are only adding this weight to make our definition of sharing incentive consistent with the assumption $\sum\limits_i w_{il} = 1$
    
    
    
    Now note that $\sum\limits_{r\in \cup_{i\in N'} g^i_l} S_r \geq \sum\limits_{i\in N'} s_{il}$ because each agent has access to her own accessible supply. By the definition of dominant resource $\frac{w_{i*}}{d_{i*}}d_{il}\leq w_{il}$. So for any $N'\subset N, l\in[L]$
    $$\left(\sum\limits_{r\in \cup_{i\in N'} g^i_l} S_r\right) /
      \left(\sum\limits_{i\in N'} w_{i*}\frac{d_{il}}{d_{i*}}\right)\geq \frac{\sum\limits_{i\in N'} s_{il}}{\sum\limits_{i\in N'} w_{il}} = 1.$$
    
    After rearranging the terms, we have 
    \begin{equation}
        \sum\limits_{r\in
    \cup_{i\in N'} g^i_l} S_r\geq
    \left(\sum\limits_{i\in N'} w_{i*}\frac{d_{il}}{d_{i*}} \right)\; \forall N'\subseteq N, l\in\{1\ldots L\}. 
    \label{eq: sharing incentive}
    \end{equation}

    For every agent $i$ and every meta-type $l$, consider $w_{i*}\frac{d_{il}}{d_{i*}}$ as the ``total demand'' for resource meta-type $l$ from agent $i$.

    Then, we construct a bipartite graph as follows: for the left-hand nodes, we
    create a node for every $\epsilon$ unit of total demand from each agent for
    each resource meta-type. Thus, each node is associated with some specific
    agent $i$ and resource meta-type $l$. For the right-hand nodes, we create a
    node for every $\epsilon$ unit of supply of each resource type ($r\in R$).
    Note that since there is a finite number of agents and resource types, there
    exists an $\epsilon$ small enough that it can perfectly divide up all the
    demands and supplies, assuming that all the weights are rational.

    Next, we create an edge between each pair of left and right-hand side nodes
    if and only if the supply side node belongs to the demand group of that
    agent for that meta-type: $r\in g^i_l$.

    Eq.\eqref{eq: sharing incentive} now tells us that for
    every subset of the demand side nodes, the number of neighbors of that
    subset is greater than or equal to the size of the subset. This is precisely
    the condition in \emph{Hall's Theorem}, which states that if this condition
    holds, then there exists a matching in the bipartite graph such
    that the demand side nodes are covered.

    Consider such a matching obtained via Hall's Theorem. We construct a
    solution $x$ by setting $x_{ir}$ equal to $\epsilon$ times the number of
    matched edges corresponding to $ir$. This yields an assignment that gives
    each agent $w_{i*}\frac{d_{il}}{d_{i*}}$ of each meta-type. By the construction of the matching this is
    a legal allocation. Then, we can set $y_t = 1$ to obtain a feasible
    solution to the optimization problem in \eqref{opt prob}.

    This means that after the first round of DRF-MT, agent $i$'s utility is at least
    \begin{equation*}
      u_i(x_i)
    = \min_{g_l \in G_i}w_{i*}\frac{d_{il}}{d_{i*}}\frac{1}{d_{il}}
    = \frac{w_{i*}}{d_{i*}} = \min\limits_{l: d_{il}\neq 0} \frac{s_{il}}{d_{il}}
    \end{equation*}

    The final expression is exactly the utility agent $i$ gets from her own supplies.
\end{proof}

\section{Experimental Setup and Additional Experiments}
\label{sec: additional experiments}
Solving for MNW is an Exponential Cone program, and solving for Discrete MNW is a Mixed Integer Exponential Cone program. Both of which did not have a reliable commercial solver until recently. This changed with the introduction of MOSEK version 9, which added support for such cones \cite{aps2020mosek}. We implemented the MNW and Discrete MNW using the MOSEK solver and our DRF-MT approach using GUROBI. We made little effort to optimize either approach beyond the off-the-shelf implementations, and all experiments are run on a 2019 16-inch Macbook Pro with a 6 core Intel i7 processor.

First we describe in more details the random instance generating procedure used in Section~\ref{experiments}. Recall that we fixed a meta-type structure. From there,
for each agent, the group structure is generated by first uniformly sample a size between $0$ and $|\Omega_l|$, and randomly pick a subset of that size from $\Omega_l$ as the demand group. The demands and
weights (before normalization) are sampled uniformly from $[1, 10]$, and the number of agents range from  $n=5$ to $n=1000$. The supply for each resource is uniformly sampled from $[n\times 500, n\times 1000]$.

\begin{figure}
    \centering
    \includegraphics[width=0.75\linewidth]{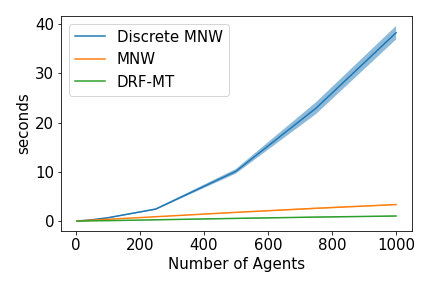} 
    \caption{Running times comparison with meta-types: $\Omega_1 = \{0\}, \Omega_2=\{1, 2\}, \Omega_3=\{3, 4, 5\}, \Omega_4=\{6, 7, 8, 9\}, \Omega_5=\{10,11,12,13,14\}$.}
    \label{fig:runtimeb} 
\end{figure}

Using the same procedure, we also compare the running times on a larger instance (more meta-types, and more types within a meta-type). As shown in Figure~\ref{fig:runtimeb}, the relative performances remain the same. 

Next, we scale up the number of meta-types instead of number of agents, even though we think it is more natural to have problems with large number of agents. Here we assume that each meta-type has five types and fix the number of agents to $50$. We see in Figure~\ref{fig:runtime meta types} again that the running time for Discrete-MNW quickly blows up while DRF-MT remains the fastest method out of the three.
\begin{figure}
    \centering
    \includegraphics[width=0.75\linewidth]{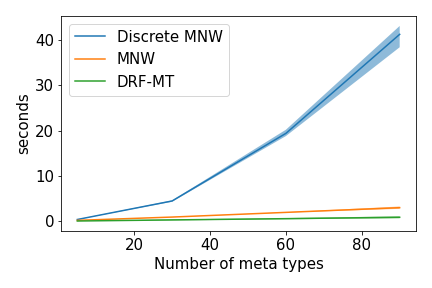}
    \caption{Running time comparison with respect to number of resources}
    \label{fig:runtime meta types}
\end{figure}

Since Discrete-MNW runs much slower than DRF-MT and MNW, we next focus on the running time comparison of just MNW and DRF-MT in our next set of experiments. The setups are the same as the ones used for Figure~\ref{fig:main experiments} (left) and 
Figure~\ref{fig:runtime meta types}, except this time we scale the problem instances to much larger ones. We see from Figure~\ref{fig:running time 2 methods} and Figure~\ref{fig:running time 2 methods mt} that DRF-MT is 3-4 times faster than MNW in terms of running time. 
\begin{figure}
    \centering
    \includegraphics[width=0.75\linewidth]{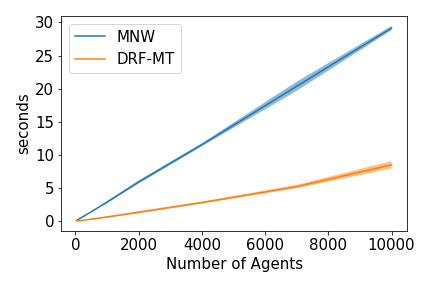}
    \caption{Running time comparison between MNW and DRF-MT. The meta types are $\Omega_1 = \{0\}, \Omega_2=\{1, 2\}, \Omega_3=\{3, 4, 5\}, \Omega_4=\{6, 7, 8, 9\}$}
    \label{fig:running time 2 methods}
\end{figure}
\begin{figure}
    \centering
    \includegraphics[width=0.75\linewidth]{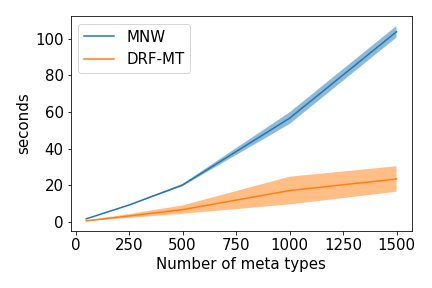}
    \caption{Running time comparison between MNW and DRF-MT with respect to number of resources.}
    \label{fig:running time 2 methods mt}
\end{figure}

Finally, as mentioned in Section~\ref{experiments}, 
normalized difference in social welfare is calculated from subtracting the social welfare of Discrete MNW from that of DRF-MT and then divide the difference by the social welfare of Discrete MNW. Here, instead of looking at the aggregated normalized difference in social welfare as in Figure~\ref{fig:main experiments} (right), we group the results by the number of agents. 
\begin{figure}[h]
      \centering
      \includegraphics[width=0.8\linewidth]{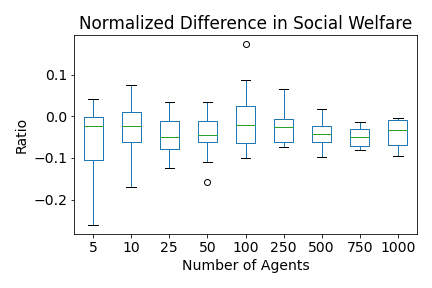} 
      \caption{Normalized difference in social welfare between Discrete MNW and DRF-MT, grouped by number of agents.} 
      \label{fig:swa} 
\end{figure}
The box plot in Figure~\ref{fig:swa} shows that in most of the trials, DRF-MT generates social welfare that is comparable to that of Discrete MNW, and sometimes even higher, with no significant variation in performance with respect to the number of agents. Readers might wonder how does Discrete DMW compare to (rounded) MNW in terms of social welfare. In Figure~\ref{fig:sw dmnw vs mnw} we see that the two have virtually the same social welfare on almost all instances, with Discrete MNW having a slight edge over MNW. This means that compared to rounded MNW, our rounded DRF-MT algorithm also achieves at least $90\%$ of the social welfare on most instances.
\begin{figure}[]
      \centering
      \includegraphics[width=0.8\linewidth]{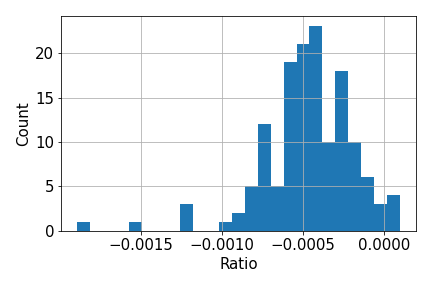} 
      \caption{Normalized difference in social welfare between Discrete MNW and MNW over all trials. Normalized difference is
    calculated by subtracting the social welfare of Discrete MNW from that of
    DRF-MT and then dividing by the social welfare of Discrete MNW.} 
      \label{fig:sw dmnw vs mnw} 
\end{figure}

\section{Beyond Sharing Incentive: Proportionality }
Proportionality is defined as follows:

\paragraph{Proportionality}
An allocation $x$ satisfies proportionality if $u_i(x_i) \geq u_i(x'_i)$ for all $i$, where $x'_{ir} = w_{il} S_r$ for each $r \in g_l^i$ and $l\in [L]$. $u_i(x')$ can be explicitly written out as 
$$u_i(x') := \min_{g_l \in G_i} \left \{\frac{w_{il}}{d_{il}} \sum_{r\in g_l}S_{r} \right \}.$$

In the existing resource allocation literature, sharing incentive and proportionality are often used interchangeably. Indeed, when the priority weights are set according to the agents' accessible contributions to the resource pool ($w_{il} = s_{il}$), the two notions are equivalent in settings where there are no accessibility constraints. With accessibility constraints, however, they are not the same. Under our definition of proportionality, the amount of accessible resource meta-type $l$ that agent $i$ receives is $s_{il}\times \sum_{r\in g_l^i}S_{r}$. Since $\sum_{r\in g_l^i}S_{r}<1$ if agent $i$ cannot access the entire supply of meta-type $l$, $u_i(x')$ can be smaller than $u_i^o$. Therefore when priority weights are set according to agents' contributions, proportionality is a weaker notion than sharing incentive. However, since proportionality can be defined for arbitrary weights, regardless of whether or not we are in a setting where agents bring their own supplies, it is a more flexible concept. 
Unfortunately proportionality does not hold generally. We prove proportionality under the following assumption:
\begin{assumption}
    \begin{align*}
    &\min\limits_{N'\subseteq N, l\in\{1\ldots L\} }
      \left(\sum\limits_{r\in \cup_{i\in N'} g^i_l} S_r\right) /
      \left(\sum\limits_{i\in N'} w_{i*} d_{il}/d_{i*}\right)\\ &\geq \max \limits_{i\in N} \left\{ \min\limits_{l: d_{il}\neq 0} 
          \frac{w_{il} \sum\limits_{r\in g_l^i} S_r }{w_{i*} d_{il}/d_{i*}}
      \right\}
      \end{align*}
    \label{assum:proportionality}
\end{assumption}

\begin{lemma}
    Assume that demands, weights and supplies are all rational numbers. Then under Assumption \ref{assum:proportionality}, DRF-MT satisfies
    proportionality.
    \label{lemma: proportionality}
\end{lemma}

The proof is very similar to that of Lemma \ref{lemma: sharing incentive}, but we first give some intuition for Assumption~\ref{assum:proportionality}.

Since $\sum\limits_{r\in g^i_l} S_r \leq 1$ for every $i, l$, and $\min\limits_{l} \frac{w_{il}}{d_{il}} = \frac{w_{i*}}{d_{i*}} $ for every $i$, the right hand side of Assumption \ref{assum:proportionality} is upper bounded by 1. $\sum\limits_{r\in \cup_{i\in N'} g^i_l} S_r$ is the union of the acceptable supply of resource meta-type $l$ from every agent in $N'$. $\sum\limits_{i\in N'} w_id_{il}/d_{i*}$ is the total weighted demand from agents in set $N'$. Note that $d_{il} \leq d_{i*}$. 

So, what the condition says intuitively is that whenever there is a group of agents who have a large combined weighted demand on meta-type $l$, they need to also collectively have access to/accept a large fraction of the total supply of $l$. 


To provide more intuition for this assumption, we look at two examples. First we check that Example \ref{ex:example} satisfies Assumption \ref{assum:proportionality}. The RHS of the assumption evaluates to $1$ (with hospital $1$ and resource meta-type $1$). One can check that the minimum on the LHS is achieved by picking $N' = N$ and $l=1$ which gives us $\frac{16}{15} > 1$. Thus Assumption \ref{assum:proportionality} is satisfied. The resulting allocation and utilities using DRF-MT is given in Table \ref{table:example 1}. Clearly our allocation is better for everyone than the proportional allocation. 

\begin{table}
\centering
\begin{tabular}{ |c|c|c|c| } 
 \hline
       \textbf{\shortstack{DRF-MT\\Allocations}}& \shortstack{Hospital 1\\($w_1= 1/4$)} & \shortstack{Hospital 2\\ ($w_2=1/4$)} & \shortstack{Hospital 3 \\($w_3=1/2$)}\\ 
       \hline
 Doctor A &  &  100 & 400\\ 
 \hline
 Doctor B & 400 &  & 100\\ 
 \hline
 Nurse  C & 100 & 400 &\\ 
 \hline
 Nurse  D &  &  & 500\\ 
 \hline
 \hline
\textbf{Utilities} & & &\\
 \hline
 DRF-MT & 100 &100 &500\\
 \hline
 Proportional & 62.5 & 31.25 &250 \\ 
 \hline
\end{tabular}
\caption{Allocations from DRF-MT in Example \ref{ex:example} and the comparison of the resulting utilities with utilities of proportional allocation.}
\label{table:example 1}
\end{table}

However, by adjusting the weights of the hospitals we can also construct an example that does not satisfy Assumption \ref{assum:proportionality}. 
Take the same parameters of~Example \ref{ex:example} with the following modification on weights: $w_1 = 0.49$, $w_2 = 0.49$, $w_3 = 0.02$. The RHS value of Assumption \ref{assum:proportionality} does not change. However, because the weights of hospitals $1, 2$ now dominate the market, the minimum of LHS is achieved with $N' = \{1,2\}$ and $l=2$, which gives us $\frac{1/2}{0.49\times 1 + 0.49\times 1/4} < 1$. So the assumption is violated. Intuitively, the problem with this setup is that even though hospitals 1 and 2 account for vast majority of the weighted demand for the nurse meta-type, they are both severely constrained to the same half of the total supply of nurses. 

Under this setup, the DRF-MT assignments/utilities do not change. With proportional allocation however, the utilities for the three agents are 
$[122.5, 61.25, 10.0]$. So Agent 1 received more utility under proportional allocation than the allocation given by DRF-MT, at the expense of significantly hurting the social welfare: the sum of the utilities is less than $200$, compared to $700$ generated by the DRF-MT allocation. Now we prove Lemma~\ref{lemma: proportionality}
\begin{proof}
    Let $\hat y$ denote the RHS of Assumption \ref{assum:proportionality}. After rearranging we have for all $N'\subseteq N$ and $l\in\{1\ldots L\}$:
    $$
    \sum\limits_{r\in
    \cup_{i\in N'} g^i_l} S_r\geq
     \hat y\left(\sum\limits_{i\in N'} w^i_* d_{il}/d_{i*} \right) 
    $$

    For every agent $i$ and every meta-type $l$, consider $\hat y w_{i*} d_{il}/d_{i*}$ as the ``total demand'' for resource meta-type $l$ from agent $i$.

    Then we construct a bipartite graph and apply Hall's theorem the same way as in the proof of Lemma~\ref{lemma: sharing incentive}. This yields an assignment that gives
    each agent at least $\hat y w_{i*} d_{il}/d_{i*}$ of each meta-type.
    By the definition of $\hat{y}$, it follows that the utility of each agent after the first round is at least:

    \begin{align*}
    \frac{d_{il}}{d_{i*}} \hat y w_{i*}\frac{1}{d_{il}}
    = \frac{w_{i*}}{d_{i*}} \hat y
    &\geq \frac{w_{i*}}{d_{i*}} \min\limits_{l: d_{il}\neq 0}\frac{w_{il}\sum\limits_{r\in g_l^i} S_r}{w_{i*}d_{il}/d_{i*}} \\ &= \min\limits_{l: d_{il}\neq 0} \frac{w_{il}\sum\limits_{r\in g_l^i} S_r}{d_{il}}
    \end{align*}
    Note that the right-most quantity is the utility of the proportional allocation.
    This means that after the first round, every
    agent already achieves at least as much utility as the proportional
    allocation. 
    Fact~\ref{fact:monotonicity of lp over rounds} finishes the proof.

\end{proof}

\section{Alternative Design of DRF-MT}
\label{sec: alternative algo design}
Instead of the algorithm described in Section~\ref{algorithm}, an alternative is to make each remaining agent receive a $y_t \times w_{i*}\times d_{il}/d_{i*}$ fraction of the total supply from \emph{each of its resource group} $g_{l_i^*}$ (instead of the supply from the entire meta-type). This idea might seem more intuitive since agent $i$ can only derive utilities from resources in $g_{l_i^*} \subseteq \Omega_{l_i^*}$. To do so, we can multiply the left hand side of the allocation constraints in Equation \ref{opt prob} by $\sum_{r\in g_l^i}S_r$. 
This alternative setup, however, does not lead to a mechanism with envy-freeness and strategy-proofness.

As a simple example, assume that there are five agents $1,2,3,4,5$ of equal weights, one meta-type, and two resource types $A, B$ that fall under this meta-type, with equal supply. Agent $1,2$ accept only type $A$; agent $3,4,5$ accept only type $B$. With simple calculation, we have that the largest $y_1$ we can get is 1/3: everyone receives $1/3$ of their accepted supply. The only possible allocation to achieve that is by assigning 1/3 of A each to agents 1,2, and 1/3 of B each to agents 3, 4, 5. However, if agent $2$ strategically stated that he could take both $A$ and $B$, the resulting allocation would be assigning 1/3 of A to agent 1, 2/3 of A to agent 2, and 1/3 of B each to agents 3, 4, 5. In this new allocation, the largest $y_1$ is still $1/3$, but since the total accepted supply for agent 2 is larger, he receives more. Furthermore, agent 1 would now envy agent 2.


\end{document}